\documentclass[11pt]{article}
\usepackage[utf8]{inputenc}
\usepackage[left=1in, right=1in, bottom=1.15in, top=1.1in]{geometry}
\usepackage{amsmath,amssymb}

\usepackage{MG}

\title{
Autobidding Equilibria in Sponsored Shopping
}

\author{
Paul Dütting\footnote{Email: duetting@google.com.}\\Google
\and Yuhao Li\footnote{Email: yuhaoli@cs.columbia.edu. Part of the work was done when the author was a student researcher at Google.}\\Columbia University
\and Renato Paes Leme\footnote{Email: renatoppl@google.com.}\\Google
\and Kelly Spendlove\footnote{Email: spendlove@google.com.}\\Google
\and Yifeng Teng\footnote{Email: yifengt@google.com.} \\Google
}

\date{}

\begin{document}

\maketitle

\begin{abstract}
As commerce shifts to digital marketplaces, platforms increasingly monetize traffic through Sponsored Shopping auctions. Unlike classic ``Sponsored Search", where an advertiser typically bids for a single link, these settings involve advertisers with broad catalogs of distinct products. In these auctions, a single advertiser can secure multiple slots simultaneously to promote different items within the same query. This creates a fundamental complexity: the allocation is combinatorial, as advertisers simultaneously win a bundle of slots rather than a single position.

We study this setting through the lens of autobidding, where value-maximizing agents employ uniform bidding strategies to optimize total value subject to Return-on-Investment (ROI) constraints. We analyze two prevalent auction formats: Generalized Second-Price (GSP) and Vickrey-Clarke-Groves (VCG). Our first main contribution is establishing the universal existence of an Autobidding Equilibrium for both settings. Second, we prove a tight Price of Anarchy (PoA) of 2 for both mechanisms.
\end{abstract}


\section{Introduction}

The landscape of commerce is undergoing a fundamental shift as online shopping displaces traditional physical retailing. Accompanying this transition is the rapid ascent of ``Sponsored Shopping". Today, virtually all major commerce platforms such as Amazon, Google Shopping, Walmart, Alibaba, and Target monetize their traffic by displaying sponsored product listings in prominent slots alongside organic results. These allocations are determined via real-time auctions triggered by user search queries.

While Sponsored Shopping shares many similarities with the classic Sponsored Search models \cite{edelman2007internet, varian2007position}, it diverges as it is fundamentally a multi-item setting. Unlike a service advertiser who might bid on a single landing page, a retailer typically possesses a vast catalog of distinct products. A smartphone vendor, for instance, may wish to promote several different models simultaneously for the query ``Android phone". This leads to complex allocation dynamics where:

\begin{itemize} \item \textbf{Multi-Item Allocation:} A single agent can win \emph{multiple} slots in the same auction to display different products.
\item \textbf{Item-Dependent Valuations:} An agent's valuation is not merely slot-dependent; it depends on the specific product-slot pair, creating a combinatorial preference structure.
\end{itemize}

We focus on two foundational auction formats: the Generalized Second-Price (GSP) auctions, used by ad auctions in major search platforms like Google Search and Microsoft Bing; and the Vickrey–Clarke–Groves (VCG) auctions, used by social network ads in Facebook (see Eg. \cite{varian2014vcg} for a discussion of comparison). In our multi-item, multi-slot setting, both auctions use the same allocation rule: we sort all items by their bids and assign the item with the $k$-th highest bid to the $k$-th slot and break ties arbitrarily. The payment rules differ for the two mechanisms. In brief, in GSP auctions \emph{without self-pricing}, for each position the winner pays the next-highest bid that is \emph{not} among their bids.\footnote{We do not allow an ad getting priced by another ad from the same advertiser to avoid self competition. This is a natural assumption adapted by major ads platforms such as Google \cite{googleads_auction} to ensure that each ad campaign does not have internal bid competition. This assumption does not make the equilibrium existence and the PoA result in our paper easier. In fact, both theorems extend to GSP that allows self-pricing.} In VCG auctions, each bidder pays the externality they impose on others. 

In studying bidding equilibria in these auctions, we consider the model of autobidding, which aligns with the structure of modern ad marketplaces.\footnote{Autobidding is the prevalent form to manage advertising campaigns, across different digital advertising products 
\cite{BergemannBW23,DengEtAl24}. 
For example, \cite{marketingltb} estimates that in 2025 over 90\% of all digital display ads were purchased programmatically.
} In the autobidding paradigm, advertisers submit high-level goals and constraints—such as a budget or a target Return on Ad Spend (tROAS)—rather than per-query bids. An algorithmic proxy then converts these constraints into specific bids at auction time, optimizing for predicted conversion rates and value in real-time.

We model this behavior using the \emph{uniform-bidding strategy}, a canonical model for value-maximizing autobidders subject to ROI constraints. Under this strategy, an agent's bid for any given item is their true value scaled by a single, uniform multiplier across all their items. The agent's goal is to maximize the total value of allocated ads subject to a strict ROI constraint (or equivalently, a tROAS constraint), ensuring that their total payment does not exceed their realized value.

We adopt the solution concept of \emph{autobidding equilibrium} \cite{li2024vulnerabilities}, where each bidder chooses a multiplier such that they either satisfy their ROI constraint exactly or hit a maximum bid cap. 

\subsection{Our Results}

Our first main theorem establishes the existence of an autobidding equilibrium in Sponsored Shopping  under both the GSP and the VCG auction format.

While proving the existence of autobidding equilibria is non-trivial even in single-item settings, the multi-item sponsored shopping environment introduces significant combinatorial complexity that renders previous techniques insufficient. 

We provide a detailed discussion of the techniques in \Cref{sec:techniques} below.

\medskip\noindent
\textbf{Main Result 1.} (\Cref{thm:existence_main}).
For any sponsored shopping auction with GSP or VCG payment, there always exists an autobidding equilibrium.
\medskip

As our second main result, we provide 
a tight analysis of the Price of Anarchy (PoA)---the worst-case ratio between the optimal social welfare and the welfare achieved in an autobidding equilibrium.
Namely, we show that for both GSP and VCG, the \emph{tight} PoA is $2$.

\medskip\noindent
\textbf{Main Result 2.} (\Cref{theorem: poa for GSP})
For any GSP (or VCG) auction, the welfare of any autobidding equilibrium is at least half of the welfare of the optimal allocation.
\medskip

We note that these PoA bounds match the bounds shown in \cite{aggarwal2019autobidding} and \cite{deng2021towards} for GSP and VCG in multi-slot auctions. 

As we discuss in \Cref{sec:techniques}, the characteristics of the sponsored shopping environment which set it apart from the multi-slot case present significant additional challenges, and necessitate the development of several new tools for bounding the Price of Anarchy. 

\subsection{Technical Contributions}
\label{sec:techniques}

We begin by examining the existence of an autobidding equilibrium. 
Previous existence results for autobidding equilibria (also called pacing equilibria in the literature)—whether under budget constraints~\cite{conitzer2022multiplicative} or ROI constraints~\cite{li2024vulnerabilities}—have primarily focused on the simple second-price auction. To address the discontinuity inherent in the second-price payment rule, these works rely on so-called $(\eps,H)$-smoothed pacing games. In this approach, they modify the allocation and payment rules of the auction itself, using the Debreu-Glicksberg-Fan theorem to establish the existence of an equilibrium in the smoothed game, and subsequently proving that the limit of a sequence of these equilibria converges to an equilibrium of the original game.

However, the smoothed mechanisms constructed in these works are explicitly tailored to the second-price auction and do not generalize easily to the more complex combinatorial structures of GSP and VCG auctions that we are considering here. Instead of modifying the mechanism, we introduce a 
\emph{generic input-smoothing framework}. We perturb the bids directly; this allows us to treat the auction mechanism as a black box, requiring only that it yields a unique outcome under strict bid orderings. This shift not only simplifies the proof and improves intuition but also makes the framework highly adaptable to other complex settings.

Our framework is as follows: We define the smoothed game by perturbing the bid for each item with a small random variable drawn uniformly and independently from $[0, \eps]$. The crucial observation is that, with probability 1, the perturbed bids induce a strict ordering of items. Consequently, the allocation and payments are uniquely determined by the auction format without ambiguity in tie-breaking. This ensures that the expected value and payment for each bidder are well-defined. In most well-behaved auction formats (in particular GSP and VCG in our paper), the expected value and payment can be shown to be continuous with respect to the multipliers. The existence of an autobidding equilibrium in the smoothed game then follows from the Brouwer fixed point theorem.

A subtle technical challenge arises regarding the ROI constraint. In the standard game, a bidder $i$ with $\alpha_i=1$ is guaranteed not to violate their ROI constraint (as they never bid above value). In the smoothed game, however, the random perturbation may cause a bidder with $\alpha_i=1$ to ``overbid'' by up to $\eps$, potentially violating the strict ROI constraint. We address this by establishing the existence of an equilibrium under a \emph{relaxed} ROI constraint in the smoothed game. We then show that as $\eps \to 0$, the limit of these equilibria converges to a state where the strict ROI constraint is satisfied.

It is also interesting to remark that recent works~\cite{filos2023fixp,filos2024ppad} have introduced powerful optimization-based techniques to establish equilibrium existence, with the goal of simultaneously settling upper bounds on complexity, namely, $\FIXP$/$\PPAD$-membership. While these techniques effectively simplify and unify proofs for a wide range of problems (including pacing/autobidding equilibria for second-price auctions under budget or ROI constraints), it is unclear whether they can be applied to our setting. These frameworks require payments to be representable by simple arithmetic circuits, whereas in GSP and VCG, the price of an item depends on the entire suffix of the global ranking (e.g., finding the highest competing bid in GSP), introducing combinatorial dependencies that are difficult to encode as simple arithmetic circuits. Our framework, therefore, offers an independent and structurally simpler pathway to establishing existence and is highly adaptable to other complex auctions.

To establish the PoA upper bound of 2, we employ a smoothness-style argument. Specifically, we prove that for any allocation $a$, the sum of social welfare and revenue is lower-bounded by the optimal social welfare:
\begin{equation}\label{eq:smoothness_condition}
    \mathrm{Wel}(a) + \mathrm{Rev}(a) \ge \mathrm{Wel}^{\mathsf{opt}}.
\end{equation}
Combined with the ROI feasibility condition, which implies that in the aggregate the expected social welfare is lower-bounded by the expected revenue ($\mathrm{Wel}(\pi) \ge \mathrm{Rev}(\pi)$), this inequality yields the factor of 2.

The core technical challenge lies in establishing inequality (\ref{eq:smoothness_condition}) for GSP and VCG. Standard smoothness proofs typically select an \emph{arbitrary} optimal allocation $a^{\mathsf{opt}}$ as a benchmark and demonstrate that the mechanism extracts sufficient revenue from ``winners'' to compensate for the value of items in $a^{\mathsf{opt}}$ that they displace. This approach often relies on the fact that payment is invariant across optimal allocations. In our setting, however, the combinatorial nature of the payment rules crucially implies that different optimal allocations may result in different aggregate payments. 

To overcome this, we show that for any allocation $a$, there exists a \emph{specific} optimal allocation $a^{\mathsf{opt}}$ (constructed based on $a$) that satisfies the bound. We prove that even though a single bidder may win multiple slots, the GSP and VCG rules extract sufficient revenue to compensate for the items displaced in \emph{this particular} optimal allocation.

This result relies on a \emph{disjoint ownership argument}: we prove that for any prefix of slots $k$, the set of items winning in the mechanism but not in the optimum, and the set of items in the optimum but not in the mechanism, must belong to disjoint sets of bidders. This property ensures that the displaced optimal items serve as valid ``price-setters'' in the payment rules, guaranteeing high enough revenue to bound the welfare loss. We believe this argument may be of independent interest for future studies of GSP and VCG in autobidding environments.

\subsection{Related Work}

\paragraph{Autobidding and PoA} Autobidding systems have become the standard in online advertising, shifting the paradigm from manual bidding to algorithmic proxy bidding under constraints.
Advertisers typically operate under a Return-on-Investment (ROI) constraint, specifying a target cost per conversion (tCPA) or target return-on-ad-spend (tROAS).  \cite{aggarwal2019autobidding} initiated the study of autobidding systems for value-maximizers, proving that uniform bidding strategies are optimal repeated truthful auctions. They further proved the existence of Pure Nash Equilibrium under uniform bidding, with the PoA being 2 for second price auctions. \cite{deng2021towards} further generalized these results to multi-slot auctions, showing that under uniform bidding the PoA defined by Pure Nash Equilibrium for both GSP and VCG is 2. For non-truthful auctions, uniform bidding may not be truthful, and \cite{liaw2023efficiency} proves that in first-price auctions the PoA is still 2. Without uniform bidding, the PoA for GSP can be as bad as 0 \cite{deng2024efficiencyGSP}. For more results related to autobidding, we refer the readers to check a thorough survey \cite{aggarwal2024auto}.

\paragraph{Autobidding Equilibrium} The autobidding equilibrium for budget constrained advertisers has been studied in various auction settings including first-price and second-price auctions \cite{balseiro2019learning,conitzer2022pacing,chen2024complexity,chen2021throttling,conitzer2022multiplicative,chen2025constant}, while the autobidding equilibrium of ROI constrained bidders is less studied. For ROI constrained autobidders, Li and Tang \cite{li2024vulnerabilities} formalized the concept of autobidding (pacing) equilibrium for ROI-constrained bidders, and demonstrated the computational hardness of (approximately) computing the optimal equilibrium.  \cite{filos2024ppad} provided a general framework proving PPAD-membership for search problems with exact rational solutions, citing ROI constrained bidders as an important application. However, their general framework cannot easily be applied to establish equilibrium existence in the multi-item setting, as it is hard to construct a gadget for identifying the winning bidders' identity within their framework. Beyond these existence results, \cite{paes2024complex} demonstrate that the existence of an autobidding equilibrium is not sufficient to guarantee that autobidding dynamics will converge.

\paragraph{Position Auction}
The theoretical framework for position auctions has been established by early works \cite{edelman2007internet,varian2007position} after the commercial success of search ads in the early 2000s. These foundational works formally defined the GSP auction used in sponsored search, and characterized a locally envy-free (or symmetric) equilibrium with the same optimal welfare as VCG. \cite{leme2010pure} initiates the study of the PoA of GSP auctions for both full-information settings and Bayesian settings, and gives the first constant PoA for these settings. The results were subsequently improved by \cite{caragiannis2011efficiency}, where the paper shows a PoA of 1.282 with pure-strategy Nash Equilibrium in the full-information setting, and a PoA of 2.927 in the Bayesian setting.

\section{Preliminaries}
In a multi-slot auction, we have $n$ bidders and $K$ slots. Each bidder $i$ has $m$ items (or ads), with values
\[
v_{i1}\geq v_{i2}\geq \cdots\geq v_{im}.
\]

The slots for this auction have decreasing importance to the bidders, which is modeled by the click-through-rate $c_1\geq c_2\geq\cdots\geq c_K$. This means if bidder $i$'s $j$th item is placed in the $k$-th slot, then the bidder $i$'s utility for this item is $c_k v_{ij}$. Without loss of generality, we assume that $v_{ij}\geq 0$ for all $i\in[n]$ and $j\in[m]$ and $c_1\leq 1$.

In this work, we are interested in two different mechanisms: the generalized second-price and VCG mechanisms. They share a common allocation rule.

\paragraph{Allocation Rule}
The mechanisms allocate items by bids. We sort all items by their bids and allocate the item with the $k$-th highest bid to the $k$-th slot. In the event of tied bids, we apply a randomized tie-breaking rule. Formally, let $\Omega$ be the set of all permutations of the $nm$ items. Let $a\in\Omega$ be an allocation where $a(k)=(i,j)$ denotes that the item at rank $k$ is bidder $i$'s $j$-th item. An allocation $a\in \Omega$ is \emph{valid under bids $\mathbf b$} if the items are ordered non-increasingly by bid:$$b_{a(1)} \ge b_{a(2)} \ge \cdots \ge b_{a(nm)}.$$Given bids $\mathbf b$ and a valid allocation $a$, the value of bidder $i$ is defined as$$V_i(a) \coloneqq \sum_{k\in[K] : a(k)=(i,j)} c_k \cdot v_{ij}.$$
We define the \emph{social welfare} of an allocation $a$ as the total value of all bidders:
\[
\mathrm{Wel}(a) \coloneqq \sum_{i=1}^n V_i(a).
\]

Let $\pi$ be a distribution over permutations in $\Omega$. We say that $\pi$ is \emph{valid under $\mathbf b$} if every allocation in the support of $\pi$ is valid under $\mathbf b$. We interpret $\pi$ as a \emph{tie-breaking rule}. Given a tie-breaking rule $\pi$, the expected value of bidder $i$ is
\[
V_i(\pi) \coloneqq \mathbb{E}_{a\sim \pi}\left[ V_i(a)\right].
\]
Accordingly, the expected social welfare under a tie-breaking rule $\pi$ is
\[\mathrm{Wel}(\pi) \coloneqq \mathbb{E}_{a\sim \pi}[\mathrm{Wel}(a)].
\]
We let $\mathrm{Wel}^{\mathsf{opt}}$ denote the optimal social welfare with respect to the true valuations:
\[
\mathrm{Wel}^{\mathsf{opt}}\coloneqq \max_{a\in \Omega} \mathrm{Wel}(a).
\]
Note that while social welfare depends only on the items assigned to the top $K$ slots, the relative ordering of items ranked below position $K$ remains relevant, as it affects payments in the GSP and VCG mechanisms that we will introduce below. 

The payment rules differ for the two mechanisms. In brief:
\begin{itemize}
\item GSP: For each position, the winner pays the next-highest bid that is \emph{not} among their bids.
\item VCG: Each bidder pays the externality they impose on others.
\end{itemize} 
We formally define the payment rules for them next.

\subsection{Generalized Second-Price (GSP) Auctions}

At a high level, in a multi-slot position auction, the generalized second-price (GSP) payment rule charges each allocated item the next highest bid.
In our setting, since each bidder may bid for multiple items, we further impose a natural \emph{no self-pricing} requirement: an item is priced only against items submitted by other bidders. Accordingly, each item pays the highest bid among all lower-ranked items submitted by other bidders.

We now formally describe the generalized second-price (GSP) payment rule for the
multi-item, multi-slot position auction.

Fix bids $\mathbf b$ and a valid allocation $a\in\Omega$ under $\mathbf b$.
For each item $(i,j)$, let $k$ denote its position in the global order $a$, i.e.,
$a(k)=(i,j)$.
If $k$ exceeds the number of available slots $K$, the item is not allocated and
incurs zero payment.
Otherwise, the GSP \emph{per-click price} for item $(i,j)$ is defined as
\[
\tau_{ij}^{\mathrm{GSP}}(a,\mathbf b)
\coloneqq
\max\Bigl\{
b_{i'j'}
\mid
(i',j') = a(k'),\ k'>k,\ i'\neq i
\Bigr\},
\]
with the convention that the maximum of an empty set is zero. Then the GSP payment for item $(i,j)$ is defined as 
\[
p_{ij}^{\mathrm{GSP}}(a,\mathbf b)\coloneqq c_k\cdot \tau_{ij}^{\mathrm{GSP}}(a,\mathbf b).
\]
The total payment of bidder $i$ is then given by
\[
P_i^{\mathrm{GSP}}(a,\mathbf b)
\coloneqq
\sum_{j=1}^m p_{ij}^{\mathrm{GSP}}(a,\mathbf b).
\]
And the expected payment of bidder $i$ under a tie-breaking rule $\pi$ is 
\[
P_i^{\mathrm{GSP}}(\pi,\mathbf b)
\coloneqq
\mathbb{E}_{a\sim \pi}\left[ P_i^{\mathrm{GSP}}(a,\mathbf b)\right].
\]

\subsection{Vickrey–Clarke–Groves (VCG) Auctions}
At a high level, the VCG mechanism charges each bidder the externality they impose on the other bidders. In our setting, this corresponds to the reduction in the total \emph{bid-welfare} of all other bidders caused by the presence of bidder $i$. Since a single bidder may win multiple slots, the payment is calculated at the bidder level rather than the item level.

We now formally describe the VCG payment rule. First, we define the \emph{bid-welfare} of an allocation $a \in \Omega$ with respect to bids $\mathbf{b}$ as the sum of the bids of the allocated items, weighted by the click-through rates:
\[
W(a, \mathbf{b}) \coloneqq \sum_{k=1}^K c_k \cdot b_{a(k)}.
\]
Note that the allocation rule selects an allocation $a$ that maximizes this quantity.

To determine the externality of bidder $i$, we consider the counterfactual scenario where bidder $i$ does not participate. Let $\mathbf{b}_{-i}$ denote the bid profile where all bids from bidder $i$ are replaced by zeros. We define the \emph{optimal counterfactual bid-welfare} without bidder $i$ as:
\[
W^{\mathsf{opt}}_{-i}(\mathbf{b}) \coloneqq \max_{a' \in \Omega} W(a', \mathbf{b}_{-i}).
\]
This term represents the maximum possible bid-welfare achievable by the remaining bidders if bidder $i$ were absent.

Now, fix the actual bids $\mathbf{b}$ and a valid allocation $a \in \Omega$ chosen by the mechanism. The bid-welfare realized by other bidders in this specific allocation is given by:
\[
W_{-i}(a, \mathbf{b}) \coloneqq \sum_{k=1}^K \mathbf{1}[a(k) \notin \{i\}\times[m]] \cdot c_k \cdot b_{a(k)}.
\]
The VCG payment for bidder $i$ under allocation $a$ is then defined as the difference between the optimal welfare of others (without $i$) and the realized welfare of others (with $i$):
\[
P_i^{\mathrm{VCG}}(a, \mathbf{b}) \coloneqq W^{\mathsf{opt}}_{-i}(\mathbf{b}) - W_{-i}(a, \mathbf{b}).
\]
And the expected payment of bidder $i$ under a tie-breaking rule $\pi$ is
\[
P_i^{\mathrm{VCG}}(\pi,\mathbf{b}) \coloneqq \mathbb{E}_{a\sim \pi}\left[ P_i^{\mathrm{VCG}}(a,\mathbf{b})\right].
\]

\subsection{An Example of GSP and VCG Payment}

We next present a simple example that illustrates the differences between the two payment rules.

\begin{example}[GSP and VCG Payments] Consider a multi-slot auction with $3$ slots. The click-through rates are $c_1 = 1.0$, $c_2=0.7$ and $c_3 = 0.5$.
\begin{itemize}
    \item Bidder 1 bids on two ads: Ad $A_1$ with bid $b_{A_1}=10$ and ad $A_2$ with bid $b_{A_2}=8$.
    \item Bidder 2 bids on two ads: Ad $B_1$ with bid $b_{B_1}=15$, and ad $B_2$ with bid $b_{B_2}=6$.
\end{itemize}

GSP and VCG produce the same allocation: As bids $b_{B_1}>b_{A_1}>b_{A_2}>b_{B_2}$, ad $B_1$ is allocated to slot 1; ad $A_1$ is allocated to slot 2; ad $A_2$ is allocated to slot 3.

\medskip
\textbf{GSP Payment (No Self-Pricing):} 
For each slot, the price is set by the highest bid from a different bidder ranked below that slot. 
\begin{itemize}
\item The price of slot 1 is determined by $b_{A_1}$ with price $c_1\cdot b_{A_1}=1.0\cdot 10 = 10$.
\item The price of slot 2 and 3 are determined by $b_{B_2}$, with the second slot having price $c_2\cdot b_{B_2}=0.7\cdot 6=4.2$ and the third slot having price $c_3\cdot b_{B_2}=0.5\cdot 6=3$.
\end{itemize}

The total payment from Bidder 1 is 7.2 for slot 2 and 3. The payment from Bidder 2 is 10 for winning the first slot.

\medskip
\textbf{VCG Payment:}
Each bidder pays the externality they impose on the other bidder. In other words, a bidder pays the difference between the total welfare others would have obtained if the bidder were absent, and the welfare they actually obtain.\footnote{The welfare in VCG mechanisms is the ``bid welfare'' that is determined by the bids reported by the agents. It is different from the ``social welfare'' that is defined from the agents' true values.}
\begin{itemize}
\item In the current allocation, Bidder 1 has allocated welfare $c_2\cdot b_{A_1}+c_3\cdot b_{A_2}=0.7\cdot10+0.5\cdot8=11$; Bidder 2 has allocated welfare $c_1\cdot b_{B_1}=1.0\cdot 15=15$.
\item When Bidder 2 is removed, Bidder 1 is allocated the first 2 slots, with welfare $c_1\cdot b_{A_1}+c_2\cdot b_{A_2}=1.0\cdot10+0.7\cdot8=15.6$. Thus Bidder 2's payment is $15.6-11=4.6$.
\item When Bidder 1 is removed, Bidder 2 is allocated the first 2 slots, with welfare $c_1\cdot b_{B_1}+c_2\cdot b_{B_2}=1.0\cdot15+0.7\cdot6=19.2$. Thus Bidder 1's payment is $19.2-15=4.2$.
\end{itemize}
\end{example}

\subsection{Autobidding Equilibria}
As discussed in the introduction, we focus on the setting with \emph{autobidders}, where every bidder participates using a uniform-bidding strategy. Specifically, each bidder $i$ chooses a multiplier $\alpha_i\in[1,A]$, and submits a bid $b_{ij}=\alpha_i \cdot v_{ij}$ for their $j$th item. 
In the existing literature, an autobidding equilibrium is typically parameterized by a pair $(\alpha, x)$, where $\alpha$ represents the bidders' multipliers and $x(\alpha)$ denotes an explicit tie-breaking rule—determined by the mechanism—that depends solely on $\alpha$. Explicitly encoding these tie-breaking rules is essential for establishing existence in general settings~\cite{conitzer2022pacing,chen2024complexity,chen2021throttling,conitzer2022multiplicative,chen2025constant,li2024vulnerabilities}, as the solution often hinges on specific boundary behaviors. This requirement mirrors a broader tradition in equilibrium theory, notably in the foundational works of~\cite{Arrow-Debreu,hylland1979efficient,bogomolnaia2001new} and many others. In this paper, we follow this established methodology and also define autobidding equilibrium as a pair $(\alpha,\pi)$, where $\pi$ is the tie-breaking rule defined above. Note that as demonstrated in \cite{li2024vulnerabilities}, the parameter $A$ below is a sufficiently large number relative to the instance parameters.

\begin{definition}[Autobidding Equilibria for $\calM$]\label{definition: pacing equilibrium GSP}
    Given a position auction $G = (n, m, K, (v_{ij}), (c_k),\calM)$, where $\calM\in\set{\mathrm{GSP},\mathrm{VCG}}$ is the payment rule, a pair $(\alpha, \pi)$, where $\alpha=(\alpha_i)\in [1,A]^n$ and $\pi$ is a distribution over all allocations $a\in\Omega$, 
  is an \emph{autobidding equilibrium} of $G$ if the following conditions hold:
\begin{itemize}
    \item[(a)] \emph{Bid-consistent allocation.}
    Every allocation in the support of $\pi$ is valid under the induced bids
    $\mathbf b$, i.e., items are ordered in non-increasing order of bids.

    \item[(b)] \emph{ROI feasibility.}
    For every bidder $i\in[n]$, the expected value obtained under $\pi$
    is at least the expected payment under $\calM$: $
    V_i(\pi) \ge P_i^{\calM}(\pi,\mathbf b).$
    That is, each bidder satisfies their return-on-investment (ROI) constraint.

    \item[(c)] \emph{Maximal autobidding.}
    For every bidder $i\in[n]$, if the ROI constraint is slack $V_i(\pi) \;>\; P_i^{\calM}(\pi,\mathbf b)$,
    then bidder $i$ is maximally paced, i.e., $\alpha_i = A$.
\end{itemize}
\end{definition}

\section{Equilibrium Existence}

In this section, we establish the universal existence of autobidding equilibria.

\begin{theorem}\label{thm:existence_main}
    For any position auction $G$ with payment rule $\mathcal{M}\in\{\mathrm{GSP},\mathrm{VCG}\}$, there exists an autobidding equilibrium $(\alpha, \pi)$ satisfying \Cref{definition: pacing equilibrium GSP}.
\end{theorem}

We prove this theorem using the smoothed framework. As discussed in the introduction, it remains an interesting open question whether the techniques developed in \cite{filos2023fixp,filos2024ppad} can be applied to establish equilibrium existence in our setting. 

In the smoothed framework, we introduce a smoothed version of the auction where bids are perturbed by small, continuous random noise. We first show that an autobidding equilibrium always exists in the smoothed game by the Brouwer fixed point theorem. Then, we recover an equilibrium for the original game by taking the limit as the noise magnitude approaches zero.

\subsection{The Smoothed Framework}
We describe the smoothed auction parameterized by $\eps>0$ below: For every bidder $i$ and item $j$, the bid is perturbed by a random noise $\xi_{ij}$ drawn independently from a uniform distribution over $[0, \eps]$. Formally, given multipliers $\boldsymbol{\alpha} \in [1, A]^n$, the perturbed bid for item $(i,j)$ is defined as
\[
\tilde{b}_{ij}(\alpha_i) \coloneqq \alpha_i v_{ij} + \xi_{ij}, \quad \text{where } \xi_{ij} \sim U[0, \eps].
\]
Let $\boldsymbol{\xi} \in [0, \eps]^{nm}$ denote the vector of perturbations. The allocation and payments are determined by the realized perturbed bids $\tilde{\mathbf{b}}$.

Let $\mathcal{M} \in \set{\mathrm{GSP}, \mathrm{VCG}}$ denote the mechanism. We define the smoothed expected value and expected payment for bidder $i$ under mechanism $\mathcal{M}$ as:
\[
\mathcal{V}^\eps_i(\boldsymbol{\alpha}) \coloneqq \mathbb{E}_{\boldsymbol{\xi}}\left[ V_i(\tilde{\mathbf{b}}(\boldsymbol{\alpha})) \right] 
\quad \text{and} \quad 
\mathcal{P}^{\mathcal{M}, \eps}_i(\boldsymbol{\alpha}) \coloneqq \mathbb{E}_{\boldsymbol{\xi}}\left[ P_i^{\mathcal{M}}(\tilde{\mathbf{b}}(\boldsymbol{\alpha})) \right].
\]
For the notations $V_i(\tilde{\mathbf{b}}(\boldsymbol{\alpha})) $ and $P_i^{\mathcal{M}}(\tilde{\mathbf{b}}(\boldsymbol{\alpha})) $ above, we omit the allocation parameter since, with probability 1, there are no ties and there is a unique valid allocation under the bids. Thus, these expected quantities are well-defined without specifying a tie-breaking rule. Indeed, since the noise terms $\xi_{ij}$ are drawn from a continuous distribution, the probability of any two distinct items having identical perturbed bids is zero. Consequently, for any fixed multipliers $\boldsymbol{\alpha}$, the perturbed bids induce a strict ordering of items with probability 1. Thus, the allocation, values, and payments are uniquely determined almost everywhere with respect to the measure of $\boldsymbol{\xi}$, making the specific choice of tie-breaking rule irrelevant to the expectations.

The equilibrium conditions (ROI feasibility and maximal autobidding) are naturally adapted to these smoothed expectations. However, in the smoothed auction, we need to relax the ROI constraint. The relaxed ROI constraint says that for every bidder $i$, the expected value obtained under $\pi$ is not too much below the expected GSP payment: 
\[
\mathcal{V}^\eps_i(\boldsymbol{\alpha}) \ge \mathcal{P}^{\calM,\eps}_i(\boldsymbol{\alpha})-\eps m.
\]

We now state our intermediate existence result.

\begin{theorem}\label{thm:smoothed_existence}
For any $\eps > 0$, there exists a vector $\boldsymbol{\alpha}^{(\eps)} \in [1, A]^n$ that constitutes a smoothed autobidding equilibrium. That is, for every bidder $i$:
\begin{itemize}
    \item $\mathcal{V}^\eps_i(\boldsymbol{\alpha}^{(\eps)}) \ge \mathcal{P}^{\calM,\eps}_i(\boldsymbol{\alpha}^{(\eps)})-\eps m$.
    \item If $\mathcal{V}^\eps_i(\boldsymbol{\alpha}^{(\eps)}) > \mathcal{P}^{\calM,\eps}_i(\boldsymbol{\alpha}^{(\eps)})$, then $\alpha_i^{(\eps)} = A$.
\end{itemize}
\end{theorem}

\begin{proof}
We define a continuous map $f: [1, A]^n \to [1, A]^n$ and apply Brouwer's fixed point theorem to find a fixed point $\boldsymbol{\alpha}^{(\eps)}$ of $f$, and show that $\boldsymbol{\alpha}^{(\eps)}$ satisfies the properties above.

At a high level, $f$ is designed to adjust the multipliers based on the bidders' ROI constraints. For each bidder $i$, let $f_i(\boldsymbol{\alpha})$ be defined as:
\[
f_i(\boldsymbol{\alpha}) \coloneqq \Pi_{[1, A]} \Big( \alpha_i + \mathcal{V}^\eps_i(\boldsymbol{\alpha}) - \mathcal{P}^{\calM,\eps}_i(\boldsymbol{\alpha}) \Big),
\]
where $\Pi_{[1, A]}(x) = \min(A, \max(1, x))$ is the truncation onto the interval $[1, A]$.

\begin{claim}\label{claim: continuous V and P}
    $\mathcal{V}^\eps_i(\boldsymbol{\alpha})$ and $\mathcal{P}^{\mathcal{M},\eps}_i(\boldsymbol{\alpha})$ are continuous functions for $\boldsymbol{\alpha}\in[1,A]^n$.
\end{claim}
\begin{proof}
    We analyze the structure of the outcome space induced by the perturbation $\boldsymbol{\xi} \sim U[0, \eps]^{nm}$.
    
    First, we show that the probability of a tie is zero for any given multipliers $\boldsymbol{\alpha}$. A tie occurs between two items $(i,j)$ and $(i',j')$ if their perturbed bids are equal:
    \[
    \alpha_i v_{ij} + \xi_{ij} = \alpha_{i'} v_{i'j'} + \xi_{i'j'}.
    \]
    This equation defines a hyperplane $H_{ij,i'j'}(\boldsymbol{\alpha})$ within the perturbation space $[0, \eps]^{nm}$.
    The set of all perturbations causing a tie, $\mathcal{T}(\boldsymbol{\alpha})$, is the finite union of such hyperplanes over all pairs of items. Thus, we have $\mathrm{Vol}(\mathcal{T}(\boldsymbol{\alpha})) = 0$.
    
    The perturbation space $[0, \eps]^{nm}$ is thus partitioned by these hyperplanes into a finite set of open polytopes. Within each polytope (region), the bids have a strict ranking. We can write the expected value as a sum over these allocations (recall that $\Omega$ is the set of all permutations of items):
    \[
    \mathcal{V}^\eps_i(\boldsymbol{\alpha}) = \sum_{a \in \Omega} V_i(a) \cdot \frac{\mathrm{Vol}(\mathcal{R}_a(\boldsymbol{\alpha}))}{\eps^{nm}},
    \]
    where $\mathcal{R}_a(\boldsymbol{\alpha})$ is the region of perturbations $\boldsymbol{\xi}$ resulting in the (unique) valid allocation $a$.
    The volume $\mathrm{Vol}(\mathcal{R}_a(\boldsymbol{\alpha}))$ is determined by the boundaries defined by the hyperplanes $H_{ij,i'j'}(\boldsymbol{\alpha})$. Since the equations defining these hyperplanes are linear in $\boldsymbol{\alpha}$, the boundaries shift continuously with $\boldsymbol{\alpha}$. Consequently, the volume of each region $\mathcal{R}_a(\boldsymbol{\alpha})$ varies continuously with $\boldsymbol{\alpha}$. As $\mathcal{V}^\eps_i$ is a linear combination of these continuous volumes, it is continuous.
    
    Similarly, the expected payment is given by:
    \[
    \mathcal{P}^{\mathcal{M},\eps}_i(\boldsymbol{\alpha}) = \sum_{a \in \Omega} \int_{\mathcal{R}_a(\boldsymbol{\alpha})} P_i^{\mathcal{M}}(a, \tilde{\mathbf{b}})\frac{1}{\eps^{nm}}\, d\boldsymbol{\xi}.
    \]
    Within each fixed ranking region $\mathcal{R}_a(\boldsymbol{\alpha})$, the allocation is fixed to $a$, and the payment function $P_i^{\mathcal{M}}$ (for both GSP and VCG) becomes a linear function of the bids (and thus of $\boldsymbol{\alpha}$). The integral is therefore continuous because both the integrand and the domain of integration (the polytope $\mathcal{R}_a(\boldsymbol{\alpha})$) deform continuously with $\boldsymbol{\alpha}$.
\end{proof}
Given \Cref{claim: continuous V and P}, since $\mathcal{V}^\eps_i$ and $\mathcal{P}^{\calM,\eps}_i$ are continuous, the map $f$ is continuous. Clearly $f$ maps $[1,A]^n$ to itself.
Thus, we can apply Brouwer's fixed point theorem to conclude that there exists a fixed point $\boldsymbol{\alpha}^{(\eps)}$ such that $f(\boldsymbol{\alpha}^{(\eps)}) = \boldsymbol{\alpha}^{(\eps)}$.

Let $\boldsymbol{\alpha}^{(\eps)}$ be a fixed point. For each bidder $i$, the update rule implies:
\[
\alpha_i^{(\eps)} = \Pi_{[1, A]} \Big( \alpha_i^{(\eps)} + \mathcal{V}^\eps_i(\boldsymbol{\alpha}^{(\eps)}) - \mathcal{P}^{\calM,\eps}_i(\boldsymbol{\alpha}^{(\eps)}) \Big).
\]
We verify the equilibrium conditions by considering the cases for the projection:
\begin{itemize}
    \item If $1 < \alpha_i^{(\eps)} < A$, the projection is inactive, implying $\alpha_i^{(\eps)} = \alpha_i^{(\eps)} + \mathcal{V}^\eps_i - \mathcal{P}^{\calM,\eps}_i$. Thus $\mathcal{V}^\eps_i(\boldsymbol{\alpha}^{(\eps)}) = \mathcal{P}^{\calM,\eps}_i(\boldsymbol{\alpha}^{(\eps)})$. In this case, the strict ROI constraint is satisfied.
    
    \item If $\alpha_i^{(\eps)} = A$, the projection implies $\mathcal{V}^\eps_i(\boldsymbol{\alpha}^{(\eps)}) - \mathcal{P}^{\calM,\eps}_i(\boldsymbol{\alpha}^{(\eps)}) \ge 0$ and thus the strict ROI constraint is satisfied.
    
    \item The case where $\alpha_i^{(\eps)} = 1$ requires a different analysis. Recall that under mechanism $\mathcal{M} \in \{\mathrm{GSP}, \mathrm{VCG}\}$, a bidder never pays more than their submitted bid for any item. In the standard auction setting, if a bidder chooses $\alpha_i=1$, their bid on each item is at most its value, ensuring that their total payment never exceeds their total value.
    
In the smoothed setting, however, even if $\alpha_i^{(\eps)} = 1$, the random perturbation may cause the bidder to bid slightly above their true value, potentially resulting in payments exceeding the realized value. This is why we introduced the relaxed ROI constraint. 
Since the perturbation is drawn from $[0, \eps]$, bidder $i$ overbids by at most $\eps$; that is, the submitted bid for each item $j$ satisfies $\tilde{b}_{ij} = v_{ij} + \xi_{ij} \leq v_{ij} + \eps$.
    
Consequently, for any valid allocation $a$ in the support of bids realized by any $\boldsymbol{\xi}$:
\[
    P_i^{\mathcal{M}}(a) \leq V_i(a) + \eps m .
\]
Taking the expectation over $\boldsymbol{\xi}$, we obtain $\mathcal{P}^{\mathcal{M},\eps}_i(\boldsymbol{\alpha}^{(\eps)}) \le \mathcal{V}^\eps_i(\boldsymbol{\alpha}^{(\eps)}) + \eps m $, satisfying the relaxed condition.
\end{itemize}
In all cases, $\boldsymbol{\alpha}^{(\eps)}$ satisfies the relaxed ROI feasibility and maximal autobidding conditions.
\end{proof}
\subsection{Proof of \Cref{thm:existence_main}}

We are now ready to prove the main theorem by taking the limit of the smoothed equilibria as the noise vanishes.

\begin{proof}[Proof of \Cref{thm:existence_main}]
    Let $\{\eps_t\}_{t=1}^\infty$ be a sequence of noise parameters such that $\lim_{t\to\infty} \eps_t = 0$. By \Cref{thm:smoothed_existence}, for each $t$, there exists a smoothed autobidding equilibrium $\boldsymbol{\alpha}^{(\eps_t)} \in [1, A]^n$. 
    The sequence $\{\boldsymbol{\alpha}^{(\eps_t)}\}$ lies in the compact metric space $[1, A]^n$. So there exists a subsequence (which we abuse the notation and also denote by $t$) that converges to a limit vector:
    \[
    \boldsymbol{\alpha}^* \coloneqq \lim_{t \to \infty} \boldsymbol{\alpha}^{(\eps_t)}.
    \]
    For each $t$, the smoothed auction induces a probability distribution over the set of allocations $\Omega$. Let $\pi^{(\eps_t)} \in \Delta(\Omega)$ be the distribution where $\pi^{(\eps_t)}(a)$ is the probability that allocation $a$ is chosen under bids $\tilde{\mathbf{b}}(\boldsymbol{\alpha}^{(\eps_t)})$ with noise $\eps_t$.
    Since the simplex $\Delta(\Omega)$ is compact, we can refine the subsequence further such that the distributions converge to a limit distribution $\pi^*$:
    \[
    \pi^*(a) \coloneqq \lim_{t \to \infty} \pi^{(\eps_t)}(a) \quad \forall a \in \Omega.
    \]
    We next show that $(\boldsymbol{\alpha}^*, \pi^*)$ is an equilibrium for the original game.
    \begin{lemma}
        $(\boldsymbol{\alpha}^*, \pi^*)$ defined above is an autobidding equilibrium.
    \end{lemma}
    \begin{proof}
        We first note that 
        \[
    \lim_{t \to \infty} \mathcal{V}_i^{\eps_t}(\boldsymbol{\alpha}^{(\eps_t)}) 
    =  V_i(\pi^*)\quad\text{and}\quad \lim_{t \to \infty} \mathcal{P}_i^{\calM,\eps_t}(\boldsymbol{\alpha}^{(\eps_t)}) 
    = P_i^{\calM}(\pi^*, \mathbf{b}^*).
    \]
    This is because for any fixed allocation $a$, the value $V_i(a)$ is constant, so we have 
    \[
    \lim_{t \to \infty} \mathcal{V}_i^{\eps_t}(\boldsymbol{\alpha}^{(\eps_t)}) 
    = \lim_{t \to \infty} \sum_{a \in \Omega} \pi^{(\eps_t)}(a) V_i(a) 
    = \sum_{a \in \Omega} \pi^*(a) V_i(a) 
    = V_i(\pi^*).
    \]

    Similarly for payments, since $\boldsymbol{\alpha}^{(\eps_t)} \to \boldsymbol{\alpha}^*$ and the noise vanishes, the realized bids in the smoothed game converge to $\mathbf{b}^* = (\alpha^*_iv_{ij})_{i\in[n],j\in[m]}$. Note that the payment $P_i^{\calM}(a, \mathbf{b})$ is a continuous function of the bids $\mathbf{b}$.
    Using the linearity of expectation and the continuity of payments:
    \[
    \lim_{t \to \infty} \mathcal{P}_i^{\calM,\eps_t}(\boldsymbol{\alpha}^{(\eps_t)}) 
    = \lim_{t \to \infty} \sum_{a \in \Omega} \int_{\mathcal{R}_a(\boldsymbol{\alpha}^{(\eps_t)})} P_i^{\mathcal{M}}(a, \tilde{\mathbf{b}})\frac{1}{\eps_t^{nm}}\, d\boldsymbol{\xi}. = \sum_{a \in \Omega} \pi^*(a) P_i^{\calM}(a, \mathbf{b}^*) 
    = P_i^{\calM}(\pi^*, \mathbf{b}^*).
    \]
    \end{proof}
Now, we verify the equilibrium conditions for $(\boldsymbol{\alpha}^*, \pi^*)$:
        \begin{enumerate}
            \item[(a)] \textbf{Bid-consistent allocation:} If $\pi^*(a) > 0$, then $a$ must be valid under $\mathbf{b}^*$. If $a$ were invalid (i.e., violated the rank order of $\mathbf{b}^*$), then there exists a sufficiently large $t$ such that for all $t'\geq t$, we have that $\boldsymbol{\alpha}^{(\eps_{t'})}$ is close to $\boldsymbol{\alpha}^*$ and $a$ would also be invalid under the perturbed bids $\tilde{\mathbf{b}}$, implying $\pi^{(\eps_t)}(a) = 0$. This contradicts $\pi^*(a) > 0$.
            \item[(b)] \textbf{ROI feasibility:} For each $t$, we have $\mathcal{V}_i^{\eps_t} \ge \mathcal{P}_i^{\calM,\eps_t} - \eps_t m $. Taking the limit as $t \to \infty$ yields $V_i(\pi^*) \ge P_i^{\calM}(\pi^*, \mathbf{b}^*)$.
            \item[(c)] \textbf{Maximal autobidding:} Suppose the ROI constraint is slack, i.e., $V_i(\pi^*) > P_i^{\calM}(\pi^*, \mathbf{b}^*)$. Then for sufficiently large $t$, we must have $\mathcal{V}_i^{\eps_t} > \mathcal{P}_i^{\calM,\eps_t}$. By the smoothed equilibrium condition, this implies $\alpha_i^{(\eps_t)} = A$. Taking the limit, we get $\alpha_i^* = A$.
        \end{enumerate}
        Thus, $(\boldsymbol{\alpha}^*, \pi^*)$ satisfies all conditions of an autobidding equilibrium.
\end{proof}

\section{Price of Anarchy}

In this section, we prove our second main result: tight PoA bounds for GSP and VCG.

\begin{theorem}[Price of Anarchy for GSP and VCG]\label{theorem: poa for GSP}
Consider a position auction $G$ with payment rule $\mathcal{M}\in\{\mathrm{GSP},\mathrm{VCG}\}$. For any autobidding equilibrium $(\alpha, \pi)$ satisfying \Cref{definition: pacing equilibrium GSP}, the expected social welfare satisfies
$$
\mathrm{Wel}(\pi) \ge \frac{1}{2} \mathrm{Wel}^{\mathsf{opt}}.
$$
Therefore, the Price of Anarchy of autobidding equilibria in GSP and VCG is at most $2$. Furthermore, this bound is tight for both mechanisms.
\end{theorem}
\begin{proof}
We prove the theorem by establishing the inequality
$$2\cdot \mathrm{Wel}(\pi) \ge \mathrm{Wel}(\pi) + \mathrm{Rev}^{\calM}(\pi) \ge \mathrm{Wel}^{\mathsf{opt}},$$
where $\mathrm{Rev}^{\calM}(\pi) = \sum_{i=1}^n P_i^{\calM}(\pi, \mathbf b)$ is the total expected revenue. We show each inequality separately below.

The first claim establishes that, in the aggregate, the expected social welfare is lower-bounded by the expected revenue.
\begin{claim}\label{claim: simpleee GSP}
    $\mathrm{Wel}(\pi) \ge \mathrm{Rev}^{\calM}(\pi)$.
\end{claim}
\begin{proof}
    From the definition of autobidding equilibrium (\Cref{definition: pacing equilibrium GSP}), the ROI feasibility condition requires that for every bidder $i$, $V_i(\pi) \ge P_i^{\calM}(\pi, \mathbf b)$.
Summing this inequality over all bidders $i\in [n]$, we obtain:
\[
\sum_{i=1}^n V_i(\pi) \ge \sum_{i=1}^n P_i^{\calM}(\pi, \mathbf b) \implies \mathrm{Wel}(\pi) \ge \mathrm{Rev}^{\calM}(\pi).
\]
\end{proof}

Next, it is clear that the second inequality $\mathrm{Wel}(\pi)+\mathrm{Rev}^{\calM}(\pi)\geq \mathrm{Wel}^{\mathsf{opt}}$ is true as long as point-wise the social welfare plus the revenue is lower-bounded by the optimal social welfare. To this end, we prove the following critical lemma and postpone its proof.

\begin{restatable}{lemma}{usefullemmaGSP}\label{lemma: useful hahaha}
    Let $\calM\in\set{\mathrm{GSP},\mathrm{VCG}}$. Fix bids ${\mathbf b}$ and a valid allocation $a\in\Omega$ under ${\mathbf b}$. There exists an optimal allocation $a^{\textsf{opt}}$ with respect to the true valuations such that 
    \begin{equation*}
    \mathrm{Rev}^{\calM}(a,\mathbf b)\geq \sum_{k=1}^K (c_k-c_{k+1})\cdot \left[\sum_{(i,j)\in S_k(a^{\textsf{opt}})\setminus S_k(a)}v_{ij}\right],
    \end{equation*}
    where $S_k(a)$ denotes the set of items assigned to the first $k$ slots in allocation $a$, and we define $c_{K+1} = 0$.
\end{restatable}

Given \Cref{lemma: useful hahaha}, we are ready to establish the following claim.
\begin{claim}
    For any allocation $a$ in the support of $\pi$, we have $\mathrm{Wel}(a) + \mathrm{Rev}^{\calM}(a) \ge \mathrm{Wel}^{\mathsf{opt}}$.
\end{claim}
\begin{proof}
Fix any allocation $a$ in the support of $\pi$. 
Recall that the welfare of allocation $a$ is $\mathrm{Wel}(a) = \sum_{k=1}^K c_k \cdot v_{a(k)}$. We can rewrite the welfare as:
\[
\mathrm{Wel}(a) = \sum_{k=1}^K (c_k - c_{k+1}) \cdot \sum_{t=1}^k v_{a(t)} = \sum_{k=1}^K (c_k - c_{k+1}) \sum_{(i,j)\in S_k(a)} v_{ij},
\]
where $S_k(a)$ is the set of items allocated to the first $k$ slots in $a$ (the same notation applies to other allocations used below).

By \Cref{lemma: useful hahaha}, there exists an optimal allocation $a^{\mathsf{opt}}$ (which depends on $a$) such that the total revenue $\mathrm{Rev}^{\calM}(a, \mathbf{b}) \coloneqq \sum_i P_i^{\calM}(a, \mathbf{b})$ satisfies:
\[
\mathrm{Rev}^{\calM}(a,\mathbf b)\geq \sum_{k=1}^K (c_k-c_{k+1})\cdot \left[\sum_{(i,j)\in S_k(a^{\mathsf{opt}})\setminus S_k(a)}v_{ij}\right].
\]
Adding $\mathrm{Wel}(a)$ and $\mathrm{Rev}^{\calM}(a, \mathbf b)$ together:
\begin{align*}
    \mathrm{Wel}(a) + \mathrm{Rev}^{\calM}(a, \mathbf b) 
    &\ge \sum_{k=1}^K (c_k - c_{k+1}) \left[ \sum_{(i,j)\in S_k(a)} v_{ij} + \sum_{(i,j)\in S_k(a^{\mathsf{opt}})\setminus S_k(a)}v_{ij} \right] \\
    &= \sum_{k=1}^K (c_k - c_{k+1}) \left[ \sum_{(i,j)\in S_k(a) \cup S_k(a^{\mathsf{opt}})} v_{ij} \right].
\end{align*}
Since valuations are non-negative, the sum over the union $S_k(a) \cup S_k(a^{\mathsf{opt}})$ is at least the sum over $S_k(a^{\mathsf{opt}})$. Thus:
\[
\mathrm{Wel}(a) + \mathrm{Rev}^{\calM}(a, \mathbf b) \ge \sum_{k=1}^K (c_k - c_{k+1}) \sum_{(i,j)\in S_k(a^{\mathsf{opt}})} v_{ij}= \mathrm{Wel}(a^{\mathsf{opt}}).
\]
Crucially, since $a^{\mathsf{opt}}$ is an optimal allocation, $\mathrm{Wel}(a^{\mathsf{opt}}) = \mathrm{Wel}^{\mathsf{opt}}$ regardless of the specific $a^{\mathsf{opt}}$ chosen by \Cref{lemma: useful hahaha}. Therefore, for every $a$ in the support of $\pi$:
\begin{equation}
    \label{eq:pointwise_bound}
    \mathrm{Wel}(a) + \mathrm{Rev}^{\calM}(a, \mathbf b) \ge \mathrm{Wel}^{\mathsf{opt}}.
\end{equation}
\end{proof}

Thus, we derive the second inequality via expectation: Taking the expectation of \eqref{eq:pointwise_bound} over $a \sim \pi$:
\[
\mathbb{E}_{a\sim \pi}[\mathrm{Wel}(a)] + \mathbb{E}_{a\sim \pi}[\mathrm{Rev}^{\calM}(a, \mathbf b)] \ge \mathrm{Wel}^{\mathsf{opt}} 
\implies \mathrm{Wel}(\pi) + \mathrm{Rev}^{\calM}(\pi) \ge \mathrm{Wel}^{\mathsf{opt}}.
\]
Combining this with \Cref{claim: simpleee GSP}, we conclude:
\[
2\cdot \mathrm{Wel}(\pi) \ge \mathrm{Wel}(\pi) + \mathrm{Rev}^{\calM}(\pi) \ge \mathrm{Wel}^{\mathsf{opt}} \implies \mathrm{Wel}(\pi) \ge \frac{1}{2} \mathrm{Wel}^{\mathsf{opt}}.
\]
This finishes the proof.
\end{proof}

To conclude the proof of \Cref{theorem: poa for GSP}, it remains only to establish \Cref{lemma: useful hahaha}. As the proof of this lemma is technical and differs between the two mechanisms, we defer it to the subsequent subsections. Before that, we first present an example establishing the tightness of the bound in \Cref{theorem: poa for GSP}.

\begin{example}[Tightness of the PoA Bound]
    \label{ex:tightness}
    Consider a setting with $n=2$ bidders and $K$ slots, where all slots have unit click-through rates ($c_k=1$ for all $k \in [K]$).
    \begin{itemize}
        \item Bidder 1 has $K$ items with valuations $v_{11}=K$ and $v_{1j}=\eps$ for $j \in \{2, \dots, K\}$, where $\eps > 0$ is a small constant.
        \item Bidder 2 has $K$ items with valuations $v_{2j}=1$ for all $j \in [K]$.
    \end{itemize}
    Consider the multipliers $\alpha_1 = 1/\eps$ and $\alpha_2 = 1$.
    The resulting bids are:
    \begin{itemize}
        \item Bidder 1 bids $b_{11} = K/\eps$ and $b_{1j} = 1$ for $j > 1$.
        \item Bidder 2 bids $b_{2j} = 1$ for all $j$.
    \end{itemize}
    Suppose ties are broken in favor of Bidder 1. Then, Bidder 1 wins all $K$ slots.

    \medskip
    \textbf{Equilibrium Verification:}
    Bidder 1's total value is $V_1 = K + (K-1)\eps$.
    In this instance, under GSP or VCG, Bidder 1's total payment is $P_1 = K$.
    The ROI constraint is satisfied (and becomes tight as $\eps \to 0$):
    \[
        V_1 = K + (K-1)\eps \ge K = P_1.
    \]
    On the other hand, Bidder 2 has both value and payment being 0. Thus, this constitutes a valid autobidding equilibrium. The equilibrium social welfare is $\mathrm{Wel}(\pi) = K + (K-1)\eps$.

    \medskip
    \textbf{Optimal Social Welfare:}
    The welfare-maximizing allocation assigns Bidder 1's first item (value $K$) to the first slot, and Bidder 2's items (value $1$) to the remaining $K-1$ slots.
    \[
        \mathrm{Wel}^{\mathsf{opt}} = K + (K-1) \cdot 1 = 2K - 1.
    \]

    \medskip
    \textbf{Price of Anarchy:}
    The ratio of the optimal welfare to the equilibrium welfare is:
    \[
        \frac{\mathrm{Wel}^{\mathsf{opt}}}{\mathrm{Wel}(\pi)} = \frac{2K - 1}{K + (K-1)\eps}.
    \]
    Taking the limit as $\eps \to 0$ and $K \to \infty$, the ratio approaches:
    \[
        \lim_{K \to \infty} \lim_{\eps \to 0} \frac{2K - 1}{K + (K-1)\eps} = \lim_{K \to \infty} \frac{2K - 1}{K} = 2.
    \]
    This demonstrates that the PoA bound of 2 is asymptotically tight.
\end{example}

\subsection{Proof of \Cref{lemma: useful hahaha}: GSP part}
In this subsection we focus on the GSP part in \Cref{lemma: useful hahaha}. We first restate it below.
\begin{lemma}
    Fix bids ${\mathbf b}$ and a valid allocation $a\in\Omega$ under ${\mathbf b}$. There exists an optimal allocation $a^{\textsf{opt}}$ with respect to the true valuations such that 
    \begin{equation*}
    \mathrm{Rev}^{\mathrm{GSP}}(a,\mathbf b)\geq \sum_{k=1}^K (c_k-c_{k+1})\cdot \left[\sum_{(i,j)\in S_k(a^{\textsf{opt}})\setminus S_k(a)}v_{ij}\right],
    \end{equation*}
    where $S_k(a)$ denotes the set of items assigned to the first $k$ slots in allocation $a$, and we define $c_{K+1} = 0$.
\end{lemma}
\begin{proof}
    Given the bid profile $\mathbf{b}$ and the valid allocation $a$ under $\mathbf{b}$, we construct the optimal allocation $a^{\textsf{opt}}$ with respect to the true valuations $(v_{ij})$ as follows: rank all items in descending order of their valuations $v_{ij}$. To handle ties where $v_{ij} = v_{i',j'}$, we break ties in favor of the item that appears earlier in the allocation order $a$. Let $a^{\textsf{opt}}$ be the allocation resulting from this ordering. Note that for the proof, it is crucial to choose $a^{\textsf{opt}}$ based on $a$.
    
    Recall that the GSP payment for the item at slot $k$, denoted by $a(k)$, is defined as $p_{a(k)}(a, \mathbf{b}) = c_k \cdot \tau_k(a, \mathbf{b})$, where $\tau_k(a, \mathbf{b})$ is the per-click price determined by the mechanism, i.e., the highest bid from a different bidder among items ranked below $k$ (in this proof, we drop the superscript of $\mathrm{GSP}$ for $\tau_{ij}^{\mathrm{GSP}}(a,\mathbf{b})$ and $p_{ij}^{\mathrm{GSP}}(a,\mathbf{b})$ as it is clear in this context). 
    
    The total revenue is the sum of payments and we can rewrite the revenue as:
    \begin{align*}
        \mathrm{Rev}^{\mathrm{GSP}}(a, \mathbf{b}) &= \sum_{k=1}^K p_{a(k)}(a, \mathbf{b}) = \sum_{k=1}^K c_k \cdot \tau_{a(k)}(a, \mathbf{b}) \\
        &= \sum_{k=1}^K (c_k - c_{k+1}) \cdot \left( \sum_{t=1}^k \tau_{a(t)}(a, \mathbf{b}) \right).
    \end{align*}
    Thus, to prove the lemma, it suffices to show that for every $k \in [K]$, the cumulative \emph{per-click prices} of the first $k$ slots cover the value of the items in $S_k(a^*)$ that are missing from $S_k(a)$:
    \[
        \sum_{t=1}^k \tau_{a(t)}(a, \mathbf{b}) \ge \sum_{(i,j) \in S_k(a^{\mathsf{opt}}) \setminus S_k(a)} v_{ij}.
    \]

    Fix an arbitrary $k$. Let $M_k \coloneqq S_k(a) \setminus S_k(a^{\textsf{opt}})$ be the set of items in the mechanism's top $k$ but not in the optimal set, and let $O_k \coloneqq S_k(a^{\textsf{opt}}) \setminus S_k(a)$ be the set of items in the optimal top $k$ but displaced from the mechanism's allocation. Note that $|M_k| = |O_k|$.
\begin{restatable}{claim}{DisjointOwnership}\label{claim: Disjoint Ownership}
For any item $(i,j)\in M_k$ and any item $(i',j')\in O_k$, it must hold that 
$i \neq i'$.
\end{restatable}
\begin{proof}
Assume for contradiction that a single bidder $i$ owns both
$(i,j)\in M_k$ and $(i,j')\in O_k$. Then by definition we know the following:

\begin{itemize}
    \item Since $(i,j)\in S_k(a)$ and $(i,j')\notin S_k(a)$, the allocation $a$
    ranks $(i,j)$ above $(i,j')$. Because within a bidder the induced bids are
    multiplicative and monotone, this implies $b_{ij} \ge b_{ij'}.$ Since $b_{ij}=\alpha_i\cdot v_{ij}$ and $b_{ij'}=\alpha_i\cdot v_{ij'}$, we know that $v_{ij} \ge v_{ij'}.$
    \item Conversely, since $(i,j')\in S_k(a^{\mathsf{opt}})$ and
    $(i,j)\notin S_k(a^{\mathsf{opt}})$, the optimal allocation
    $a^{\mathsf{opt}}$ ranks $(i,j')$ above $(i,j)$, which implies $v_{ij'} \ge v_{ij}.$
\end{itemize}

Combining the two inequalities, it must be that $v_{ij'} = v_{ij}$. Then by our tie-breaking rule, $a^{\mathsf{opt}}$ must break ties consistently with $a$. Since the allocation $a$ ranks $(i,j)$ above $(i,j')$, the same must hold in $a^{\mathsf{opt}}$, contradicting $(i,j')\in S_k(a^{\mathsf{opt}})$ and $(i,j)\notin S_k(a^{\mathsf{opt}})$. Thus, no bidder can simultaneously
own an item in $M_k$ and an item in $O_k$, and the two owner sets are disjoint.
\end{proof}

Now consider any item $(i,j)\in M_k$. Under the GSP rule (with no self-pricing), the price of the item
$(i,j)$ is the highest bid among all items that have rank strictly below $(i,j)$ in the allocation $a$ and are owned by other bidders.

By \Cref{claim: Disjoint Ownership}, every item $(i',j')\in O_k$ satisfies the condition above.
Thus each $(i',j')\in O_k$ is a valid price-setter for each $(i,j)\in M_k$.

Let $\bar{b}=\max_{(i',j')\in O_k} b_{i'j'}.$ Then for every $(i,j)\in M_k$, we have $\tau_{ij}(a,\mathbf b) \ge \bar{b}.$

Summing over all $(i,j)\in M_k$ gives
\[
\sum_{(i,j)\in M_k} \tau_{ij}(a,\mathbf b)
\ge
|M_k|\, \bar{b}
\;=\;
|O_k|\, \bar{b}
\;=\;
\sum_{(i',j')\in O_k} \bar{b}\geq \sum_{(i',j')\in O_k} b_{i'j'}.
\]

Since multipliers satisfy $\alpha_{i'}\ge 1$ for all $i'$, we have
$b_{i'j'} = \alpha_{i'} v_{i'j'} \ge v_{i'j'}$. Therefore,
\[
\sum_{(i,j)\in M_k} \tau_{ij}(a,\mathbf b)
\geq \sum_{(i',j')\in O_k} v_{i'j'}=\sum_{(i',j') \in S_k(a^{\mathsf{opt}}) \setminus S_k(a)} v_{i'j'}.
\]

Finally, since all prices are nonnegative,
we have
\[
\sum_{t\in [k]} \tau_{a(t)}(a,\mathbf b)
\geq \sum_{(i',j') \in S_k(a^{\mathsf{opt}}) \setminus S_k(a)} v_{i'j'}.
\]
the total payment from the top $k$
slots exceeds (or equals) the payments from the subset $M_k$, completing the
argument.
\end{proof}

\subsection{Proof of \Cref{lemma: useful hahaha}: VCG part}
In this subsection, we focus on the VCG part in \Cref{lemma: useful hahaha}.
\begin{lemma}
\label{lemma:VCG_smoothness}
Fix bids $\mathbf{b}$ and a valid allocation $a \in \Omega$ under $\mathbf{b}$. There exists an optimal allocation $a^{\mathsf{opt}}$ with respect to the true valuations such that
\[
\mathrm{Rev}^{\mathrm{VCG}}(a, \mathbf{b}) \ge \sum_{k=1}^K (c_k - c_{k+1}) \cdot \left[ \sum_{(i,j) \in S_k(a^{\mathsf{opt}}) \setminus S_k(a)} v_{ij} \right],
\]
where $S_k(a)$ denotes the set of items assigned to the first $k$ slots in allocation $a$, and we define $c_{K+1}=0$.
\end{lemma}
\begin{proof}
    We construct $a^{\mathsf{opt}}$ exactly as in the previous proof of \Cref{lemma: useful hahaha}: rank items by true valuations $v_{ij}$, breaking ties consistently with the bid-based allocation $a$.

    We first rewrite the revenue.
    
    The VCG revenue is the sum of payments: $\mathrm{Rev}^{\calM}(a, \mathbf{b}) = \sum_{i} P_i^{\mathrm{VCG}}(a,\mathbf{b})$. 
    By definition, the payment for bidder $i$ is the externality imposed on others:
    \[
    P_i^{\mathrm{VCG}}(a,\mathbf{b}) = W^{\mathsf{opt}}_{-i}(\mathbf{b}) - W_{-i}(a, \mathbf{b}).
    \]

    Recall the definition of bid-welfare for bidders other than $i$ in allocation $a$:
    \[
    W_{-i}(a, \mathbf{b}) = \sum_{k=1}^K c_k \cdot \left(  \mathbf{1}[a(k) \notin \{i\}\times[m]] \cdot b_{a(k)} \right).
    \]
    We can rewrite it as follows:
    \[
    W_{-i}(a, \mathbf{b}) = \sum_{k=1}^K (c_k - c_{k+1}) \left[ \sum_{(u,j) \in S_k(a)} \mathbf{1}[u \neq i] \cdot b_{uj} \right],
    \]
    where the inner partial sum represents the total bid value of items in the top $k$ slots of $a$ that do \emph{not} belong to bidder $i$.
    
    Similarly, let the set $S_k(a_{-i})$ be the optimal set of top $k$ items in $a$ when $i$ is excluded. The optimal bid-welfare without $i$ is:
    \[
    W^{\mathsf{opt}}_{-i}(\mathbf{b}) = \sum_{k=1}^K (c_k - c_{k+1}) \left[ \sum_{(u,j) \in S_k(a_{-i})} b_{uj} \right].
    \]
    
    Substituting these into the payment equation, we group terms by the layer weight $(c_k - c_{k+1})$:
    \[
    P_i^{\mathrm{VCG}}(a,\mathbf{b}) = \sum_{k=1}^K (c_k - c_{k+1}) \underbrace{\left[ \sum_{(u,j) \in S_k(a_{-i})} b_{uj} - \sum_{(u,j) \in S_k(a), u \neq i} b_{uj} \right]}_{\Delta_k(i)}.
    \]
    
    Next, we analyze the term $\Delta_k(i)$. 
    Recall that the set $S_k(a_{-i})$ is the optimal set of top $k$ items when $i$ is excluded. Since allocation is based on bid ranking, this set is constructed by taking the items in $S_k(a)$ that are \emph{not} owned by $i$, and filling the vacancies with the highest-bidding items from outside $S_k(a)$.
    Let $R_k(i) \coloneqq S_k(a_{-i}) \setminus S_k(a)$ be this set of ``replacement items''. Then:
    \[
    S_k(a_{-i}) = \{ (u,j) \in S_k(a) \mid u \neq i \} \cup R_k(i).
    \]
    Substituting this set relationship into $\Delta_k(i)$, the terms in $\{ (u,j) \in S_k(a) \mid u \neq i \}$ cancel out perfectly, leaving only the value of the replacement items:
    \[
    \Delta_k(i) = \sum_{(u,j) \in R_k(i)} b_{uj}.
    \]
    Thus, the total revenue is:
    \begin{equation}
    \label{eq:revenue_structure}
    \mathrm{Rev}^{\calM}(a, \mathbf{b}) = \sum_i P_i^{\mathrm{VCG}}(a,\mathbf{b}) = \sum_{k=1}^K (c_k - c_{k+1}) \left[ \sum_i \sum_{(u,j) \in R_k(i)} b_{uj} \right].
    \end{equation}
    Now we claim that:
    \begin{claim}\label{claim: VCG haha}
    For any $k$,
        \[
    \sum_i \sum_{(u,j) \in R_k(i)} b_{uj}\geq \sum_{(i,j) \in S_k(a^{\mathsf{opt}}) \setminus S_k(a)} b_{ij}.
    \]
    \end{claim}
    Given \Cref{claim: VCG haha} and $b_{ij}\geq v_{ij}$ for all $i$ and $j$, we can conclude that 
    \begin{align*}
       \mathrm{Rev}^{\calM}(a, \mathbf{b}) &= \sum_{k=1}^K (c_k - c_{k+1}) \left[\sum_i \sum_{(u,j) \in R_k(i)} b_{uj}\right] \\
        &\geq \sum_{k=1}^K (c_k - c_{k+1}) \left[ \sum_{(i,j) \in S_k(a^{\mathsf{opt}}) \setminus S_k(a)} b_{ij} \right]\\
        &\geq \sum_{k=1}^K (c_k - c_{k+1}) \left[ \sum_{(i,j) \in S_k(a^{\mathsf{opt}}) \setminus S_k(a)} v_{ij} \right].
    \end{align*}
    Therefore, it remains only to prove \Cref{claim: VCG haha}.

    \begin{proof}[Proof of \Cref{claim: VCG haha}]
        Fix an arbitrary $k$, and recall the sets defined in \Cref{lemma: useful hahaha}: $M_k = S_k(a) \setminus S_k(a^{\mathsf{opt}})$ and $O_k = S_k(a^{\mathsf{opt}}) \setminus S_k(a)$. Note that $|M_k| = |O_k|$. 
    
    We further decompose the set $M_k$ by bidder. Let $M_k(i)$ be the items in $M_k$ owned by bidder $i$, and let $m_i = |M_k(i)|$. Let $U = ([n]\times[m]) \setminus S_k(a)$ be the set of items not in the mechanism's top $k$. 
    
    When bidder $i$ is removed from the auction to calculate $W_{-i}^{\mathsf{opt}}$, all her items are removed. The removal of $i$ creates at least $m_i$ vacancies in the top $k$ (corresponding to the items in $M_k(i)$). These vacancies must be filled by items from outside $S_k(a)$ that are \emph{not} owned by $i$.
    Thus, $R_k(i)$ contains the $|R_k(i)|$ items with the highest bids from the set $U_{-i} \coloneqq U \setminus \{ \text{items owned by } i \}$, where $|R_k(i)| \ge m_i$.

    Now, we invoke \Cref{claim: Disjoint Ownership} again.
    The claim guarantees that the bidders owning items in $M_k$ are distinct from those owning items in $O_k$. 
    This implies that for any $i$ such that $m_i>0$, we have $O_k\cap(\{i\}\times[m])=\emptyset$. In other words, $O_k \subseteq U_{-i}$.

    Since $\sum_i m_i = |M_k| = |O_k|$, we can arbitrarily partition the set $O_k$ into disjoint subsets $\{Z_i\}_i$ such that for each $i$, $|Z_i| = m_i$. 
     
    For each bidder $i$, consider the set $Z_i$. If $m_i>0$ and $Z_i\neq\emptyset$, then $Z_i \subseteq O_k \subseteq U_{-i}$, and $R_k(i)$ contains the highest bidding items from $U_{-i}$ (and $|R_k(i)| \ge m_i = |Z_i|$), the total bid of replacement items must be at least the bid of items in $Z_i$:
    \[
    \sum_{(u,j) \in R_k(i)} b_{uj} \ge \sum_{(u,j) \in Z_i} b_{uj}.
    \]

    Summing over all bidders $i$:
    \[
    \sum_i \sum_{(u,j) \in R_k(i)} b_{uj} \ge \sum_i \sum_{(u,j) \in Z_i} b_{uj} = \sum_{(u,j) \in \bigcup_i Z_i} b_{uj} = \sum_{(u,j) \in O_{k}} b_{uj},
    \]
    which concludes the proof of claim.
    \end{proof}
    This finishes the proof.
\end{proof}

\section{Conclusion and Future Directions}

We established that the equilibrium properties of multi-item sponsored shopping closely parallel those of the classical single-item setting. Despite the combinatorial nature of the allocation, where agents win bundles of slots rather than single positions, the Price of Anarchy for both GSP and VCG remains exactly 2, mirroring the tightness results found in standard single-item autobidding models.  

Two significant theoretical challenges remain open. The first concerns the computational complexity of computing an autobidding equilibrium in the sponsored shopping setting. While we have established the existence of autobidding equilibria via a smoothed framework, this proof does not yield any upper bound on the complexity of computing such an equilibrium. If these equilibria turn out to be intractable, a natural direction is to settle their complexity, for instance by establishing completeness for $\FIXP$ or $\PPAD$. Another natural question is the complexity of computing equilibria that approximately optimize social welfare or revenue.

The second challenge is the question of learning dynamics. As demonstrated by Paes Leme et al. \cite{paes2024complex}, the existence of an autobidding equilibrium is not sufficient to guarantee that standard learning algorithms will converge to it. Understanding the long-term behavior and convergence properties of learning algorithms in the multi-item shopping environment remains a key direction for future research.

\section*{Acknowledgment}
The authors would like to thank Roozbeh Ebrahimi and Balasubramanian Sivan for discussion on the shopping auction model.

\begin{flushleft}
\bibliographystyle{alpha}
\bibliography{ref}
\end{flushleft}

\end{document}